\documentclass{article}
\usepackage{dannychen-reg}

\usepackage[dvipsnames]{xcolor}
\usepackage{hyperref}
\hypersetup{
  breaklinks=true,
  colorlinks=true,
  allcolors=BlueViolet,
}

\usepackage{biblatex}
\bibliography{ref}

\DeclareMathOperator{\size}{size}
\DeclareMathOperator{\qdeg}{qdeg}
\DeclareMathOperator{\std}{std}
\usepackage{tikz}

\renewcommand{\red}{{}}

\usetikzlibrary{decorations.pathreplacing,decorations.pathmorphing}
\definecolor{babyblue}{rgb}{0.54, 0.81, 0.94}
\definecolor{lavenderblue}{rgb}{0.8, 0.8, 1.0}

\begin{document}

\title{Quantum Circuit Cutting for Classical Shadows}

\author[1]{Daniel T. Chen}
\author[1]{Zain H. Saleem}
\author[2]{Michael A. Perlin}
\affil[1]{Mathematics and Computer Science Division, Argonne National Laboratory}
\affil[2]{Infleqtion, Inc.}

\date{}

\maketitle

\begin{abstract}
%Shadow tomography refers to the process of inferring properties of quantum states, such as expectations of an observable, through \textit{classical shadows}, an ensemble of classical projections of the desired quantum states.
%Although provably efficient in sample complexity, simulating large circuits to collect classical shadows remains difficult.
Classical shadow tomography is a sample-efficient technique for characterizing quantum systems and predicting many of their properties. Circuit cutting is a technique for dividing large quantum circuits into smaller fragments that can be executed more robustly using fewer quantum resources. We introduce a divide-and-conquer circuit cutting method for estimating the expectation values of observables using classical shadows. We derive a general formula for making predictions using the classical shadows of circuit fragments from arbitrarily cut circuits\red{, and provide the sample complexity analysis for the case when observables factorize across fragments}. Then, we numerically show that our divide-and-conquer method outperforms traditional uncut shadow tomography when estimating high-weight observables that act non-trivially on many qubits, and discuss the mechanisms for this advantage.
\end{abstract}

\section{Introduction}

The complexity of fully describing a quantum system grows exponentially with its size.
However, it may be considerably easier to estimate certain properties of a quantum state.
To leverage this idea, Aaronson introduced the concept of \textit{shadow tomography} \cite{aaronson2019shadow}.
Aaronson proved that some functions of a quantum state, namely expectation values of a fixed but arbitrary set of observables, can be estimated arbitrarily well using only polynomially many copies of the quantum state.
However, Aaronson's procedure required capabilities that are well out of reach for near-term experiments and quantum computing devices.
Inspired by the idea of estimating properties rather than states themselves, Huang, Kueng, and Preskill recast Aaronson's method into \textit{classical shadow tomography}, which retains the original method's favorable complexity while alleviating hardware constraints \cite{huang2020predicting, elben2022randomized, huang2022learning}.
On a high level, Huang \textit{et al.}'s shadow tomography builds a classical model of a quantum state, or a \textit{classical shadow}, that can be used to estimate the desired observables.

Classical shadow tomography, including its randomized \cite{huang2020predicting, elben2022randomized, huang2022learning}, derandomized \cite{huang2021efficient, stricker2022experimental}, and Bayesian \cite{lukens2021bayesian} variants, has subsequently received a great deal of attention.
Classical shadows have been shown to be beneficial for machine learning \cite{huang2022foundations, huang2022quantum}, avoiding barren plateaus \cite{sack2022avoiding}, estimating fermionic states \cite{zhao2021fermionic}, and mitigating errors \cite{seif2022shadow}.
Classical shadows are also provably robust to various types of noise \cite{huang2022foundations, koh2022classical}, which demonstrates compatibility with noisy near-term quantum computing devices.
The channel-state duality also allows the techniques of classical shadows to be used for efficient quantum process tomography \cite{kunjummen2021shadow, levy2021classical}.

%On a different note, quantum hardware that are currently available have low qubit counts, rendering promising applications moot. 
Meanwhile, inspired by fragmentation methods for molecular simulation \cite{warshel1976theoretical, gordon2012fragmentation, li2007generalized, li2008fragmentation}, Peng \textit{et al.} recently introduced the idea of \textit{circuit cutting} \cite{peng2020simulating}.
Circuit cutting divides a large quantum circuit into smaller subcircuits, or \textit{fragments}.
These fragments can be simulated independently and subsequently recombined by classically post-processing their measurement results.
Circuit cutting can thereby reduce the number of qubits required to simulate a quantum circuit.
This technique is also sometimes referred to as \textit{circuit knitting}.

Generally speaking, circuit cutting comes at the cost of classical post-processing overheads that are exponential in the number of cuts that are made to a circuit.
However, these overheads might be avoided with suitably-designed algorithms that use circuit cutting for distributed quantum combinatorial optimization \cite{saleem2021quantum}.
Moveover, recent studies have demonstrated that circuit cutting allows quantum circuits to be more reliably executed, and that circuit cutting can be beneficial even if a quantum device has enough qubits to simulate the entire circuit \cite{ayral2020quantum, ayral2021quantum, perlin2021quantum}.
More elaborate circuit cutting procedures have also been proposed. For example, maximum likelihood methods have been applied to mitigate the corrupting effect of hardware errors and shot noise on recombined circuit outputs \cite{perlin2021quantum}.
Circuit cutting can also be sped up by incorporating randomized measurements \cite{lowe2022fast} or stochastic classical postprocessing \cite{chen2022approximate}.
Other work has proposed to fragment a circuit by cutting gates rather than qubit wires \cite{piveteau2022circuit}. 

In this work, we introduce a divide-and-conquer circuit cutting method for estimating the expectation values of observables from the classical shadows of circuit fragments.
With the circuit-cutting formalism, we express the expectation values at the end of a circuit in terms of the classical shadows of its fragments.
This formalism inherits the flexibility and low sample complexity of classical shadows, while allowing for the simulation of smaller quantum circuits.
We further show that circuit cutting enables estimating higher-weight observables (that is, observables that act non-trivially on more qubits) with a fixed number of measurement samples.
In total, we find that if a quantum circuit is amenable to cutting, doing so increases the weight of the observables that can be accurately estimated using classical shadow tomography.

The organization of the paper is as follows. We review circuit cutting in Section \ref{sec:circuit-cut}, and derive a general formula for using fragment data to estimate expectation values with respect to the output state of a fragmented circuit. In Section \ref{sec:shadow-tomography}, we review shadow tomography, a derandomized variant thereof, and the use of classical shadows for quantum process tomography.
Section \ref{sec:divide-and-conquer} combines circuit cutting and shadow tomography, providing theoretical lower bounds on the sample complexity of estimating observables using the classical shadows of fragments, as well as numerical results comparing shadow tomography with and without fragmentation.
%the two previous sections that features both mathematical and empirical analysis of quantum divide and conquer for classical shadows. Specifically, we calculate the sample complexity of the divide-and-conquer method. Moreover, we examine the numerical results comparing shadow tomography with and without fragmentation.

\subsection{Notational Remarks}
Since the paper will be concerned with quantum circuits, we will uniformly index each qubit wire starting from $1$. 
Numerical subscripts will be reserved for actions on qubits specified by an index. 
For example, $U_{12}$ will be a quantum gate $U$ acting on qubits $1$ and $2$. 
Similarly, we will use numerical indices to denote partial traces, such that $\tr_1(\rho)$ denotes the partial trace of $\rho$ with respect to qubit 1.

For an operator $A$, $A^\intercal$ will be used to denote the matrix transpose and $A^\dagger$ will be used to denote its adjoint operator (conjugate transpose).
Sets are usually denoted with script-characters like $\mathcal S$ or greek letters. 
Complements will be denoted by $\overline{\mathcal S}$. 
However, if set complements ever occurs in subscripts, slashes will be used for cleaner typesetting, \textit{e.g.,} $A_{\not{\mathcal S}}$. 
Moreover, for a positive integer $N$, the set $[N]$ is the set that enumerates from $1$ to $N$, \textit{i.e.,}~$[N] = \{1, 2, \dots, N\}$. 

\section{Quantum circuit cutting} \label{sec:circuit-cut}
Consider three-qubit quantum state prepared by a circuit $\rho = V_{23} U_{12} \ket{000}\bra{000} U_{12}^\dagger V_{23}^\dagger$ for some arbitrary two-qubit quantum gates $U$ and $V$ respectively acting on qubits $1,2$ and $2,3$. 
%Let $\mathcal B$ be an orthonormal basis for the vector space of self-adjoint $2 \times 2$ matrices, such that $\tr(A B) = 2$ if $A=B$ and $0$ otherwise for all $A, B \in \mathcal B$. 
%In the rest of the paper, we will use the Pauli matrices for concreteness, $\mathcal B = \{I, X, Y, Z\}$. 
Let $\mathcal B = \{I, X, Y, Z\}$ be the set of Pauli operators including the identity.
The orthogonality of $\mathcal{B}$ with respect to the trace inner product implies that $\rho$ can be expanded into a sum of tensor products as \cite{peng2020simulating, perlin2021quantum}
\begin{align}\label{eq:cut-example}
    \rho &= \frac{1}{2} \sum_{M \in \mathcal B} \tr_2 \left(M_2 U_{12} \ket{00} \bra{00} U_{12}^\dagger \right) \tensor V_{23} (M_2 \tensor \ket{0}\bra{0}) V_{23}^\dagger,
\end{align}
where $M_2$ denotes the action of $M$ on qubit 2.
The two tensor factors inside the summation essentially {(almost; clarified below)} form two quantum states that can be prepared independently.
Thus, we say that the quantum circuit can be \textit{cut} into two \textit{fragments}, $f_1$ and $f_2$.
The state of each fragment is parameterized by the operator $M \in \mathcal B$, satisfying:
\begin{align}
    \rho_{f_1}(M) &= \tr_2 \left(M_2 U_{12} \ket{00} \bra{00} U_{12}^\dagger \right), \\
    \rho_{f_2}(M) &= V_{23} (M_2 \tensor \ket{0}\bra{0}) V_{23}^\dagger.
\end{align}
Careful readers should note that $\rho_{f_i}(M)$ is not actually a ``state'' since $M$ is a (possibly traceless) operator rather than a quantum state. To resolve this concern, we interpret $\rho_{f_j}$ as a linear combination of \textit{conditional states} by expanding $M = r M\ind r + s M \ind s$, where $(r, s) = \lambda(M)$ are the eigenvalues of $M$, and $M\ind j$ is a rank-1 projector. Then, we complete the conditional-state interpretation of Eq.~\eqref{eq:cut-example} with the expansion
\begin{align}\label{eq:cut-example-eig}
    \rho = \frac{1}{2} \sum_{\substack{M \in \mathcal B, \\ r,s \in \lambda(M)}} rs \cdot  \rho_{f_1}\left( M \ind{r} \right) \tensor \rho_{f_2} \left( M \ind{s} \right).
\end{align}
In the following text, we will not make the distinction between Eqs.~\eqref{eq:cut-example} and \eqref{eq:cut-example-eig} for notational cleanliness.

One can think of circuit cutting as performing tomography at the cut locations. 
We will call the qubit wire directly upstream of the cut a \textit{quantum output}, which is connected to a \textit{quantum input} directly downstream of the cut. 
Circuit cutting is performed by inferring all quantum degrees of freedom at the cut locations---summing over an informational-complete (IC) set of measurement bases at quantum outputs, and over IC set of states (such as the SIC-POVM states \cite{renes2004symmetric}) at quantum inputs. 
On the other hand, the \textit{circuit output} of a fragment are the qubit wires that were at the end of the uncut circuit. 
The ultimate goal is to combine data collected from fragments in a way that reproduces desired properties of the full, uncut circuit.

In this paper, we focus on reconstructing the expectation value of an observable $O$ with respect to the state $\rho$ prepared by a circuit, \textit{i.e.,}~$\tr(O \rho)$. 
For concreteness, suppose $O$ admits the decomposition $O = O_1 \tensor O_{23}$, and that the circuit preparing $\rho$ is amenable to cutting as in Eq.~\eqref{eq:cut-example}.
Then, we can write the expectation value as a sum of products of expectations with respect to operators that are local to each fragment:
\begin{align}
    \tr(O \rho) &= \frac{1}{2} \sum_{M \in \mathcal B} \tr \left[ (O_1 \tensor O_{23}) \cdot (\rho_{f_1}(M) \tensor \rho_{f_2}(M) )\right] \\
    %&= \frac{1}{2} \sum_{M} \tr \left[ (O_1 \rho_{f_1}(M)) \tensor (O_{23} \rho_{f_2}(M)) \right] \\
    &= \frac{1}{2} \sum_{M} \tr [O_1 \rho_{f_1}(M)] \cdot \tr[O_{23} \rho_{f_2}(M)].
    \label{eq:toycut}
\end{align}

We can reformulate the above equation in a way that abandons the distinction between initialization and measurements.
Each fragment can be viewed as a quantum channel that maps a  state at the quantum input of a fragment to a state at all outputs (which includes both the quantum and circuit outputs).
{
By the channel-state duality \cite{de1967linear, jamiolkowski1972linear, choi1975completely}, we can write the quantum channel as a tripartite \textit{Choi state} $\Lambda \in \mathcal L(\mathcal H^{Q_i} \tensor \mathcal H^{Q_o} \tensor \mathcal H^{C_o})$, where $Q_i$, $Q_o$, and $C_o$ are respectively the number of quantum inputs, quantum outputs, and circuit outputs on a fragment, $\mathcal H^N$ denotes the Hilbert space of $N$ qubits, and $\mathcal L(A)$ denotes the space of linear operators on $A$.
{Specifically, the Choi state $\Lambda$ of a fragment $f$ with $Q_i$ quantum inputs is the density matrix obtained by preparing $Q_i$ Bell pairs (\red{$(\ket{00} + \ket{11})/\sqrt{2}$}), and inserting one qubit from each of the Bell pairs ($Q_i$ qubits total) into the quantum input of fragment $f$.}
In terms of these Choi states, the expectation value of an operator $O$ can be written as the following:
\begin{align}
    \tr(O\rho) = \sum_{M} \tr (({1}\otimes M \tensor O_1)\Lambda_{f_1}) \cdot \tr((M^\intercal \tensor {1} \tensor O_{23}) \Lambda_{f_2}),
    \label{eq:cut-example-choi}
\end{align}
where for clarity we include an explicit $1$ at a tensor factor of $\mathcal L(\mathcal H^{Q_i} \tensor \mathcal H^{Q_o} \tensor \mathcal H^{C_o})$ that addresses a 1-dimensional Hilbert space, which occurs when one of $Q_i$, $Q_o$, or $C_o$ is empty (that is, when a fragment has no quantum inputs, quantum outputs, or circuit outputs).
Henceforth omitting dimension-1 tensor factors for brevity, the Choi states in the above example are
\begin{align}
  \Lambda_{f_1} = U \ket{00}\bra{00} U^\dagger,
  &&
  \Lambda_{f_2} = (I \tensor V) (\ket\Phi\bra\Phi \tensor \ket0\bra0) (I \tensor V^\dagger),
\end{align}
where $I$ is the single-qubit identity operator, and $\ket\Phi = (\ket{00} + \ket{11})/\sqrt{2}$.}

% \red{
% [TEXT LEFT FOR REFERENCE]
% In this particular example, the Choi states take the form:
% \begin{align}
%     \Lambda_{f_1} &= \ket 0 \bra 0 \tensor (U_{12} \ket{00}\bra{00}U_{12}^\dagger), \\
%     \Lambda_{f_2} &= (Id \tensor Id \tensor V_{23}) \left[ \frac{1}{2}(\ket{000}\bra{000} + \ket{101}\bra{101}) \tensor \ket 0 \bra 0 \right] (Id \tensor Id \tensor V_{23}^\dagger).
% \end{align}}
Note the distinction between a Choi \textit{matrix} and a Choi \textit{state}. 
Here, we invoke the state interpretation, which requires a normalization factor that scales with the dimension of the quantum input. 
Thus, the factor of $1/2$ that we see in Eq.~\eqref{eq:toycut} is absorbed into $\Lambda_{f_2}$ {(buried in $\ket\Phi\bra\Phi$)}, and no normalization is needed for $\Lambda_{f_1}$ because it lacks quantum inputs.
{Altogether, no factors of 1/2 are necessary when expressing expectation values in terms of Choi states, as in Eq.~\eqref{eq:cut-example-choi}, because every factor of 1/2 comes from the resolution of the identity operator at a cut, and the same factor of 1/2 is accounted for by the normalization of the Choi state with a quantum input at the same cut.}
% \red{[TEXT LEFT FOR REFERENCE]
% Furthermore, when one of the Hilbert spaces has dimension one, e.g., quantum inputs for fragment $1$ and quantum outputs for fragment $2$, we will omit it from the expression for simplicity. In fact, this has already been done for the circuit input state, which is assumed to be the all-zero state.}
% \red{To demonstrate, nelecting spaces of dimension one, Eq.~\ref{eq:cut-example} now becomes 
% \begin{align}
%     \tr(O\rho) = \sum_{M} \tr ((M \tensor O_1)\Lambda_{f_1}) \cdot \tr((M^\intercal \tensor O_{23}) \Lambda_{f_2}) \label{eq:cut-example}
% \end{align}
% where the Choi states are now
% \begin{align}
%     \Lambda_{f_1} &= U_{12} \ket{00}\bra{00}U_{12}^\dagger, \\
%     \Lambda_{f_2} &= (Id \tensor V_{23}) \left[ \frac{1}{2}(\ket{00}\bra{00} + \ket{11}\bra{11}) \tensor \ket 0 \bra 0 \right] (Id \tensor V_{23}^\dagger).
% \end{align}}

% The above example demonstrates a simple instance of quantum circuit cutting.
% Using this procedure, we can estimate expectations on each fragment separately, perhaps in a distributed manner, and classically recombine the results. 
% At the same time, the amount of data one needs to collect scales exponentially with the number of cuts because an additional set of basis $\mathcal B$ becomes part of the summation for each cut. 
% Thus, while circuit cutting techniques reduces the amount of qubits need to be allocated simultaneously, which is particularly beneficial in the NISQ era, the amount of classical resources needed might outweigh the initial interest of circuit cutting.

\subsection{Generalizing circuit cutting}\label{sec:gencirccut}

%\red{In this subsection, we aim to derive a generalization of Eq.~\eqref{eq:toycut} and \eqref{eq:cut-example} for arbitrary circuit cutting schemes. The final result (Eq.~\eqref{eq:gencutchoi} and Fig.~\ref{fig:graph_partition}) will be used in later sections.
%While the notation introduced here will be required to make sense of later results, the details of the derivation are not essential.
%We therefore provide a table of definitions for quick reference in Table \ref{tab:notation}, using which readers are free to skip ahead for a high-level reading our this paper.
%}

A circuit can be cut at arbitrary locations so long as it forms distinct fragments that can be operated on. 
Tang \textit{et al.}~proposed methods for optimizing the cut location based on post-processing complexity \cite{tang2021cutqc}. 
We will consider the case where the cut locations have already been decided.
Fragments and their connectivity can then be represented by a directed graph with fragments as vertices and qubit wires connecting fragments as edges.

More formally, an $N_q$-qubit quantum circuit is split up into fragments $ F = \{f_i\}_{i=1}^{N_f}$ with $K$ cuts. 
The collection of fragments forms a directed multigraph $G = (F, E)$ where each edge $(f_i, f_j, \ell) \in E$ represents a cut on qubit $\ell$ separating fragment $f_i$ from fragment $f_j$.
For each fragment $f_i$, there is an associated conditional quantum state $\rho_{f_i}$ (of its circuit outputs) that depends on some auxiliary operators $M \in \mathcal B$ (at its quantum inputs and outputs).
These operators specify state preparation (quantum input) and measurement (quantum output) routines for preparing the conditional state $\rho_{f_i}$.
Lastly, $\mathcal Q_i(f_i), \mathcal Q_o(f_i), \mathcal C_o(f_i)$ denote the number of quantum inputs, quantum outputs, and circuit outputs associated with fragment $f_i$.

Given an arbitrary circuit split into $N_f$ fragments $\{f_i\}$ with $K$ cuts, the quantum state prepared by the circuit can be written (up to a permutation of qubit tensor factors) as
\begin{align}
    \rho = \frac{1}{2^K} \sum_{\bm M \in \mathcal B^K} \bigotimes_{i=1}^{N_f} \rho_{f_i}(M_{f_i}),
\end{align}
where $M_{f_i}$ denotes the operators in $\bm M$ that are associated with the cuts incident to fragment $f_i$. 
So, $M_{f_i}$ is a tuple of $\mathcal Q_i(f_i) +  \mathcal Q_o(f_i)$ operators in $\mathcal{B}$. 
It is also important to note that we will allow fragments with no circuit outputs, in which case $\rho_{f_i}(M_{f_i})$ is simply a scalar.
%The dimension of the tensor product of all fragments will be consistent with the dimension of the full circuit. 

\begin{figure*}
\centering
\begin{subfigure}{0.48\linewidth}
\includegraphics[width=\linewidth]{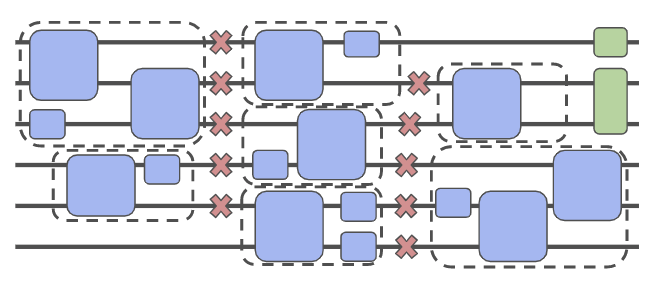}
\caption{Above is an example of circuit cutting. Blue boxes represent arbitrary quantum gates, red crosses represent cut locations, and dashed lines outline individual fragments. The two green boxes at the end of the circuit represent observables that we wish to find the expectation of.}
\end{subfigure}
\hfill
\begin{subfigure}{0.48\linewidth}
\includegraphics[width=\linewidth]{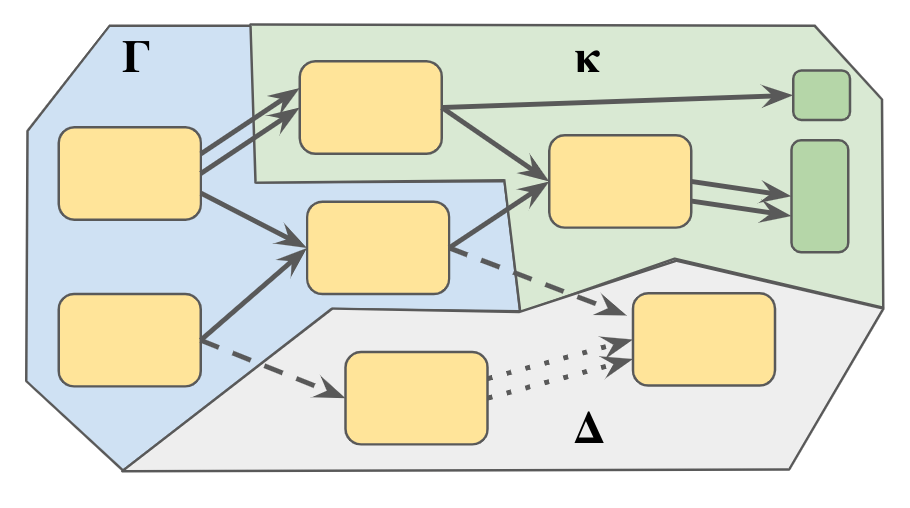}
\caption{Given the circuit and cut on the left, we can construct a fragment graph as depicted above. Each yellow box corresponds to one fragment. Each shaded region represents one part in the fragment graph partition. The green region ($\kappa$) denotes the fragments directly acted on by the observable. The blue region ($\Gamma$) denotes the fragment that are upstream of $\kappa$. The grey region denotes the fragments in $\Delta$, which does not contribute to the expectation value of the desired observable.}
\end{subfigure}
\caption{An example of circuit cutting and its mapping to a  fragment graph. {Note that the fragment graph in the example above is acyclic, but this property is not required in general.  We will later consider structured circuit ansatz whose fragment graph contains cycles.}}
\label{fig:graph_partition}
\end{figure*}

Consider an operator $O$ that is local to a subset $\mathcal K \subset [N_q]$ of the qubits that are addressed by a circuit. 
Let $\kappa$ be the set of fragment \red{labels such that fragment $f_k$ for $k \in \kappa$} have circuit outputs on the qubits in $\mathcal K$. 
Furthermore, let $O \ind \kappa$ be the restriction of $O$ onto the circuit outputs of fragments in $\kappa$; by definition, $O$ acts trivially on all other fragments.
We can separate expectation values of $O$ into factors involving the fragments that are addressed by $O$, namely $\kappa$, and all other fragments, \red{$\overline \kappa = [N_f] \setminus \kappa$}:
\begin{align}
    \tr(O\rho) &= \frac{1}{2^K} \sum_{\bm M \in \mathcal B^K} \tr\left( O \ind \kappa \bigotimes_{\red{k \in \kappa}} \rho_{f_k}(M_{f_k})\right) \cdot \prod_{\red{i \in \overline \kappa}} \tr \left( \rho_{f_i}(M_{f_i}) \right).
    \label{eq:cut_fragments_full}
\end{align}
{Note that we do not assume that $O$\red{, and by extension, $O \ind \kappa$,} factorizes across fragments.}
It is important to note that these expressions assume that tensor factors are permuted to the appropriate ordering, which will be an implicit assumption throughout this work.

We now seek to eliminate the dependence of $\tr(O\rho)$ on fragments that are outside the past light cone of $O$.
To this end, we further partition the set $\overline{\kappa}$, as visualized in Figure \ref{fig:graph_partition}.
Let $\Gamma$ be the set of fragments strictly upstream of $\kappa$ in $G$ (not including fragments in $\kappa$), and let \red{$\Delta = [N_f] \setminus(\kappa\cup\Gamma)$} be the set of all remaining fragments in $G$.
{That is, $\Delta$ is the set of fragments strictly outside the past light cone of $O$.
By noting that the expectation value $\tr(O\rho)$ is identical to that with respect to the state prepared by a circuit in which all fragments in $\Delta$ are eliminated, we can simplify}
%In terms of this decomposition, we have
%\begin{align}
%  \tr(O\rho) = \frac{1}{2^K} \sum_{\bm M \in \mathcal B^K} \tr \left( O \ind \kappa \bigotimes_{f_k \in \kappa} \rho_{f_k}(M_{f_k}) \right) \cdot \prod_{f_i \in \Gamma} \tr(\rho_{f_i} (M_{f_i})) \cdot \prod_{f_j \in \Delta} \tr(\rho_{f_j}(M_{f_j})).
%  \label{eq:three_parts}
%\end{align}
%Partition of the edge set $E$ of $G$ into three parts:
%(a) the edges $E_\Delta$ in $\Delta$,
%(b) the edges $E_{\not\Delta}$ in the complement $\overline\Delta=\kappa\cup\Gamma$, and
%(c) the edges $E_{\not\Delta\to\Delta}$ from $\overline\Delta$ to $\Delta$.
%We can sum over all operators in Eq.~\eqref{eq:three_parts} that are associated with edges in 
%$E_\Delta$ to contract all fragments in $\Delta$ into a single fragment $f_\Delta$:
%\begin{align}
%  \tr(O\rho) = \frac{1}{2^{K - |E_\Delta|}} \sum_{\bm M \in \mathcal B^{K - |E_\Delta|}} \tr \left( O \ind \kappa \bigotimes_{f_k \in \kappa} \rho_{f_k}(M_{f_k}) \right) \cdot \prod_{f_i \in \Gamma} \tr(\rho_{f_i} (M_{f_i})) \cdot \tr(\rho_{f_\Delta}(M_{f_\Delta})).
%\end{align}
%Here $M_{f_\Delta}$ contains only operators associated with the edges $E_{\not\Delta\to\Delta}$.
%Viewing $\rho_{f_\Delta}$ as a trace-preserving quantum channel from its quantum inputs to its circuit outputs, we observe that $\tr(\rho_{f_\Delta}(M_{f_\Delta}))=0$ unless all operators in $M_{f_\Delta}$ have nonvanishing trace, which only occurs if they are all the identity operator, in which case $\tr(\rho_{f_\Delta}(M_{f_\Delta}))=2^{|E_{\not\Delta\to\Delta}|}$.
%Altogether,
\begin{align}
    \tr(O\rho) = \frac{1}{2^{|E_{\not \Delta}|}} \sum_{\bm M \in \mathcal B^{|E_{\not \Delta|}}} \tr \left( O \ind \kappa \bigotimes_{\red{k \in \kappa}} \rho_{f_k}(M_{f_k}) \right) \cdot \prod_{\red{i \in \Gamma}} \tr(\rho_{f_i} (M_{f_i})),
    \label{eq:reduction}
\end{align}
%where we used the fact that $K = |E_\Delta| + |E_{\not\Delta}| + |E_{\not\Delta\to\Delta}|$, and for brevity we implicitly assign identity operators to the quantum outputs of $\rho_{f_k}$ and $\rho_{f_i}$ that are associated with edges in $E_{\not\Delta\to\Delta}$ (which is equivalent to simply tracing over these quantum outputs).
%Note that $\tr(O\rho)$ is always independent of fragments in $\Delta$.
{where $E_{\not\Delta}$ is the set of edges between the fragments in past light cone of $O$, or the complement $\overline\Delta =\kappa\cup\Gamma$.
Note that Eq.~\eqref{eq:reduction} implicitly assigns identity operators to the quantum outputs of $\rho_{f_k}$ and $\rho_{f_i}$ that are associated with edges between $\Delta$ and $\overline\Delta$, which is equivalent to simply tracing over these quantum outputs.
}

The above expression assumes that the $O\ind\kappa$ can be measured directly, which may not always be the case.
Expanding $O \ind \kappa$ in the basis of Pauli strings as $O \ind \kappa = \sum_{P \in \mathcal{B}^{\tensor |\mathcal K|}} \alpha_P P$, we find that
\begin{align}
    \tr(O \rho)
    &= \sum_{P \in \{I,X,Y,Z\}^{\tensor |\mathcal K|}} \frac{\alpha_P}{2^{|E_{\not \Delta}|}} \sum_{\bm M \in \mathcal B^{|E_{\not \Delta}|}} \prod_{\red{k \in \kappa}} \tr \left( P_{\mathcal C_o(f_k)} \rho_{f_k}(M_{f_k}) \right) \cdot \prod_{\red{i \in \Gamma}} \tr(\rho_{f_i} (M_{f_i})),\label{eq:gencutstate}
\end{align}
where $P_{\mathcal C_o(f_k)}$ is the restriction of $P$ to the circuit outputs of fragment $f_k$.
The choice of Pauli-string expansion is for generality. Practically, the cost of obtaining the Pauli string expansion grows exponentially with respect to the size of $O$. 
One can employ alternative expansions so long as the decomposed observables are separable with respect to the circuit outputs of the fragments. 

We can further apply the channel-state duality to each factor of the expansion in Eq.~\eqref{eq:gencutstate}. 
Let $M_{f_i}^{Q_i}$ be a tensor product of operators at the quantum inputs of fragment $f_i$, and similarly $M_{f_i}^{Q_o}$ for the quantum outputs. 
Then, we can write the above summation in terms of Choi states as
\begin{multline}
    \tr(O \rho) = \sum_{P \in \{I,X,Y,Z\}^{\tensor |\mathcal K|}} \alpha_P \sum_{\bm M \in \mathcal B^{|E_{\not \Delta}|}} \prod_{\red{k \in \kappa}} \tr \left(\left( {M_{f_k}^{Q_i}}^\intercal \tensor M_{f_k}^{Q_o} \tensor P_{\mathcal C_o(f_k)} \right) \Lambda_{f_k} \right) \\ 
    \cdot \prod_{\red{i \in \Gamma}} \tr \left( \left( {M_{f_i}^{Q_i}}^\intercal \tensor M_{f_i}^{Q_o} \tensor I_{\mathcal C_o(f_i)} \right) \Lambda_{f_i} \right),
    \label{eq:gencutchoi}
\end{multline}
where $I_{\mathcal C_o(f_i)}$ is the identity operator on all circuit outputs of fragment $f_i$.
As in Eq.~\eqref{eq:cut-example-choi}, normalizing factors of $2$ get absorbed into Choi states, since there are exactly $|E_{\not \Delta}|$ quantum inputs between the fragments in $\kappa\cup\Gamma$.
% Thus, the factor of $1/2^{E_{\not \Delta}}$ present in Eq.~\eqref{eq:gencutstate} is being used to normalize the Choi matrices.

To summarize, we presented a general formula for finding the expectation of an observable in terms of expectations of operators on individual fragments.
We can see that the formula sums over exponentially many terms with respect to the number of cuts. 
This is generally unavoidable without making further assumptions on a specific circuit ansatz. 
For an observable that is local to certain qubits, only the fragments that are directly acted on and the fragments upstream of those (that is, fragments in the past light cone of the observable) need to be considered. 
This decreases the exponential factor by a constant, which would correspond to an improvement in computational efficiency in practice.

\section{Shadow tomography for states and processes}
\label{sec:shadow-tomography}

Quantum tomography is the process reconstructing an unknown quantum state from measurements. 
The traditional method requires exponentially many measurements \cite{o2016efficient}. 
Though expensive, tomography is useful when the quantity to be estimated is not known a priori.
\textit{Classical shadow tomography} offers a computationally efficient way of estimating polynomially many observables with sub-exponential sample complexity \cite{huang2020predicting}. 
Intuitively, shadow tomography aims at reconstructing properties of a quantum state by taking ``snapshots'' in random bases. 
%For observables that are local to only a few qubits, random measurements sufficiently explore the low dimensional subspace and produces high quality estimates. 

Formally, we shall begin by defining the \textit{measurement primitive} $\mathcal U$---the set of measurement bases to sample from. 
% The measurement primitive and properties of desired observables together dictate the number of measurements needed. 
The measurement primitive can be arbitrary so long as it is \textit{tomographically complete}, \textit{i.e.,}~for all quantum states $\sigma \neq \rho$, there exists a $U \in \mathcal U$ and $\ket{b}$ such that
\begin{align}
    \bra b U \sigma U^\dagger \ket b \neq \bra b U \rho U^\dagger \ket b
\end{align}
In this paper, we will employ the \textit{random Pauli basis} measurement primitive:
\begin{align}
    \mathcal U = {\langle H, S \rangle^{\tensor n}} = \text{Cl}(2)^{\tensor n}
\end{align}
{where $H$ is the Hadamard gate and $S = e^{-i\pi Z/4}$ is the S gate, which together generate the group} of Clifford gates $\text{Cl}(2)$.
{Here $\mathcal{S}^{\otimes n}$ is the set of $n$-fold tensor products of elements of the set $\mathcal{S}$.}
Equivalently, one can think of this measurement primitive as measuring {each qubit separately} in a randomly-chosen Pauli basis. 
This measurement primitive has the property that the number of measurements needed to estimate expectation values of observables will scale exponentially only with the number of qubits that the observable acts on non-trivially, which we define as the size, or weight, of an observable $O$, and denote $\size(O)$. 

The general shadow tomography procedure, taken from \cite{huang2020predicting}, is as follows. After executing a circuit that prepares the state $\rho$, we pick $U \in \mathcal U$ uniformly at random, rotate $\rho \mt U \rho U^\dagger$, and measure all qubits in the computational basis. Suppose we get the measurement result $\ket{b}$. Then, we define a new operator $\mathcal M$ by
\begin{align}
    \mathcal M(\rho) = \E_{U \sim \mathcal U} \sum_{b \in \{0,1\}^n} U^\dagger \ket b \bra b U \bra b U \rho U^\dagger \ket b.
\end{align}
From here, we build an estimator for $\rho$ by inverting $\mathcal M$,
\begin{align}
    \hat \rho = \mathcal M\inv \left( U^\dagger ~ \ket{\hat b}\bra{\hat b} U \right),
\end{align}
where the inverse of $\mathcal M$ depends on the measurement primitive $\mathcal U$, and $\hat\rho$ is equal to $\rho$ in expectation: $\E \hat \rho = \rho$.
This inverse is well-defined because $\mathcal U$ is tomographically complete.
This estimator is the \textit{classical shadow} of $\rho$, and generally speaking a classical shadow can be built by averaging over many such estimators. 
%We can take an ensemble of shadows and use them to estimate some quantity of interest. 
For more robust estimation, as suggested in Ref.~\cite{huang2020predicting}, one could perform median-of-means---separating the shadow into groups of single-shot estimators, finding the mean estimate of each group, and taking the median of those means. Below we will give a theorem for the sample complexity of estimating linear functions. 

\begin{lemma}[Theorem \red{S1 in the supplementary material of} Ref.~\cite{huang2020predicting}]\label{thm:shadowtomo}For the random Pauli measurement primitive $\mathcal U$, collection of Hermitian $O_1, \dots, O_M \in \C^{2^n \times 2^n}$, error $\epsilon, \delta \in (0,1)$. Let
\begin{align}
    K = 2\log \frac{2M}{\delta},
    &&
    N = \frac{34}{\epsilon^2} \max_{i \in [M]} 4^{\size(O_i)} \|O_i\|_{\infty}^2.
\end{align}
Then, a collection of $NK$ independent classical shadows builds the estimator $\hat o_i$ for the expectation of $O_i$ satisfying
\begin{align}
    \vert \hat o_i - \tr(O_i \rho) \vert < \epsilon
\end{align}
for all $i \in [M]$ with probability at least $1-\delta$.
\end{lemma}

To obtain estimates of polynomially many observables, one would need only logarithmically growing numbers of measurements. 
Meanwhile, the number of measurements respects locality (which here means the number of qubits addressed by an operator, and has nothing to do with locality in physical space). 
For observables that are local to a few qubits, relatively few samples will be needed. 
Lastly, the choice of observables is not required a priori so long as there is an estimate for the maximum locality of observables.

\subsection{Implementation of shadow tomography}\label{sec:st-code}
Here we present an implementation that is equivalent to shadow tomography scheme above \cite{bergholm2018pennylane, huang2021efficient}. 
Suppose we want to estimate the expectation of a Pauli observable $O$. 
The idea is to generate an ensemble of bitstrings measured in randomly chosen Pauli basis and collect only those that matches $O$ at non-trivial qubits.
The estimator is then generated by averaging over the bitstrings where $O$ is directly observed. 

Procedurally, for an $n$-qubit circuit we first generate an ensemble of length-$n$ \textit{trit-strings} $s_1, s_2, \dots, s_S$, where each $s_j\in\{X, Y, Z\}^n$ specifies a measurement basis for each qubit. 
For each trit-string $s_j$, we run the desired circuit and measure in basis specified by the string, obtaining a bitstring $m_j\in\{+1,-1\}^n$ that specifies the measurement outcome for each qubit.
The collection of all trit-strings and corresponding measurement outcomes, $\{(s_j,m_j)\}_{j=1}^S$, is the classical shadow of the circuit. 
To estimate the expectation value of a Pauli observable $O$ that acts non-trivially on qubits $\mathcal{K}$, we find a corresponding trit-string $s_j$ that matches $O$ at qubits $\mathcal K$, and take a product over the associated measurement outcomes in $m_j$ at $\mathcal K$ for a single-shot estimate of $O$. 
Lastly, we take an average of these single-shot estimates to get the final estimate for the expectation value of $O$.

For example, suppose there is a $3$-qubit circuit. 
First, we generate the bases in which to measure each qubit in each independent evaluation of the circuit, and we get 
\begin{align}
    (X,Y,X), (Z,Y,Y), (X,Z,Y), (X,Y,Z), (X,X,X).
   %(1,2,1), (3,2,2), (1,3,2), (1,2,3), (1,1,1).
\end{align}
Running the circuit yields an ensemble of bitstrings like the following:
\begin{align}
    (1,1,-1), (-1,-1,1), (-1,1,-1), (-1,1,1), (1,-1,1).
\end{align}
Suppose we want to estimate the expectation value of $O = XYI$. 
We then look for occurrences of $(X,Y,\cdot)$---corresponding to $XYI$---and collect only the parts of the classical shadow that matches this pattern. 
We end up with bitstrings $(1,1,-1)$ and $(-1,1,1)$. 
Again, only the measurement outcome of the first two qubits matter, so we get a final estimate of $((1 \times 1) + (-1 \times 1))/2 = 0$. 
On the other hand, if $O = IYI$, then our estimate becomes $(1 - 1 + 1)/3= 1/3$. 
Of course, there is a chance that the corresponding Pauli string is simply not observed, \textit{e.g.,} if $O = YXI$. 
In this case, there is simply no data that allows us to return an informed estimate and, as a default, we return the uniform prior on $O$, namely $\tr(O)=0$ (for nontrivial Pauli observables).
As we will see in Section \ref{sec:numerics}, this phenomenon is a significant drawback of shadow tomography, which can be greatly mitigated with circuit fragmentation.

\subsection{Shadow process tomography}

Process tomography refers to inferring a quantum channel, a completely-positive trace-preserving map between Hilbert spaces, from measurements.
For a fixed number of qubits, process tomography is in general exponentially more expensive than state tomography because we are required to map each input---a quantum state---to its output---another quantum state. 
By the channel-state duality, the problem of process tomography can be reduced to performing state tomography on the Choi state, from which we can directly apply the results from shadow tomography \cite{kunjummen2021shadow, levy2021classical}.

Suppose we wish to know about a quantum channel $\mathcal E$ with Choi state $\Lambda \in \C^{2^{m+n} \times 2^{m+n}}$ {where $m$ and $n$ are the number of inputs and outputs respectively of the channel}. 
Suppose $\rho$ is the input state and $\sigma$ is the output state, then by channel-state duality, {the expectation of the state through a quantum channel with respect to $\sigma$} can be thought of as performing a measurement on the Choi state,
\begin{align}
    {\tr(\mathcal E(\rho) \sigma)} = 2^m \tr\left( (\rho^\intercal \tensor \sigma) \Lambda \right).
\end{align}
Define the operator $O = \rho^\intercal \tensor \sigma$. 
Like in shadow tomogrpahy, we can have a collection of operators $O$, \textit{i.e.,}~$O_1, O_2, \dots, O_M$ each specifying an input-output relation we wish to predict. 

% \begin{figure}
%     \centering
%     \includegraphics[width=.6\linewidth]{figure/shadowprocess.png}
%     \caption{An example circuit for performing shadow process tomography. The orange block depict the inputs to the quantum channel. To apply measurements in random bases, we entangle the ancilla qubits (last two qubits) with the quantum input. The layer of green represents randomly chosen Pauli bases, which are applied both on the output of the channel directly and to the input through entangled layers.}
%     \label{fig:shadowprocess}
% \end{figure}

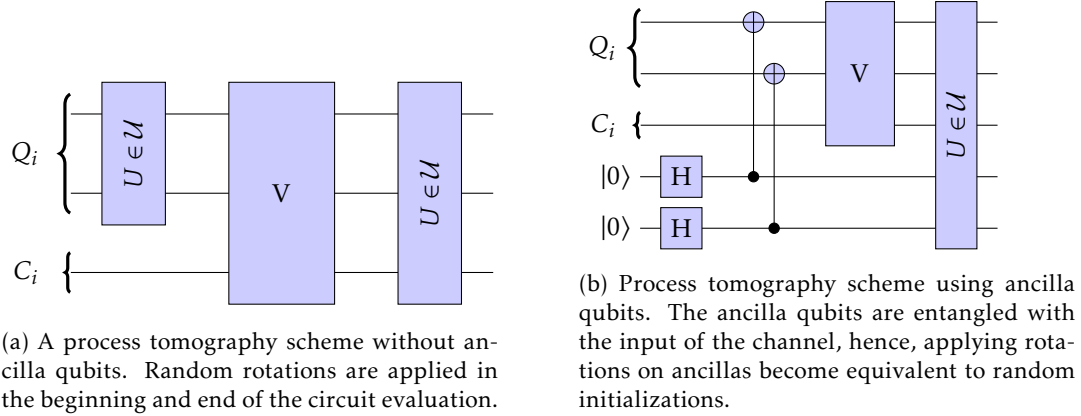
\begin{figure}
\centering
~\hfill
\begin{subfigure}{0.4\linewidth}
%! \usetikzlibrary{decorations.pathreplacing,decorations.pathmorphing}
\begin{tikzpicture}[scale=2.00000,x=1pt,y=1pt]
\filldraw[color=white] (0.000000, -7.500000) rectangle (80.000000, 37.500000);
% Drawing wires
% Line 2: a b W Q_i<
\draw[color=black] (0.000000,30.000000) -- (80.000000,30.000000);
%   Deferring wire label at (0.000000,30.000000)
% Line 2: a b W Q_i<
\draw[color=black] (0.000000,15.000000) -- (80.000000,15.000000);
\filldraw[color=white,fill=white] (0.000000,11.250000) rectangle (-4.000000,33.750000);
\draw[decorate,decoration={brace,amplitude = 4.000000pt},very thick] (0.000000,11.250000) -- (0.000000,33.750000);
\draw[color=black] (-4.000000,22.500000) node[left] {$Q_i$};
% Line 3: c W C_i<
\draw[color=black] (0.000000,0.000000) -- (80.000000,0.000000);
\filldraw[color=white,fill=white] (0.000000,-3.750000) rectangle (-4.000000,3.750000);
\draw[decorate,decoration={brace,amplitude = 1.875000pt},very thick] (0.000000,-3.750000) -- (0.000000,3.750000);
\draw[color=black] (-4.000000,0.000000) node[left] {$C_i$};
% Done with wires; drawing gates
% Line 5: a b G \rotatebox{90}{$U\in\mathcal{U}$}
\draw (12.000000,30.000000) -- (12.000000,15.000000);
\begin{scope}
\draw[fill=lavenderblue] (12.000000, 22.500000) +(-45.000000:8.485281pt and 19.091883pt) -- +(45.000000:8.485281pt and 19.091883pt) -- +(135.000000:8.485281pt and 19.091883pt) -- +(225.000000:8.485281pt and 19.091883pt) -- cycle;
\clip (12.000000, 22.500000) +(-45.000000:8.485281pt and 19.091883pt) -- +(45.000000:8.485281pt and 19.091883pt) -- +(135.000000:8.485281pt and 19.091883pt) -- +(225.000000:8.485281pt and 19.091883pt) -- cycle;
\draw (12.000000, 22.500000) node {\rotatebox{90}{$U\in\mathcal{U}$}};
\end{scope}
% Line 6: a b c G:width=20 \rotatebox{90}{V}
\draw (40.000000,30.000000) -- (40.000000,0.000000);
\begin{scope}
\draw[fill=lavenderblue] (40.000000, 15.000000) +(-45.000000:14.142136pt and 29.698485pt) -- +(45.000000:14.142136pt and 29.698485pt) -- +(135.000000:14.142136pt and 29.698485pt) -- +(225.000000:14.142136pt and 29.698485pt) -- cycle;
\clip (40.000000, 15.000000) +(-45.000000:14.142136pt and 29.698485pt) -- +(45.000000:14.142136pt and 29.698485pt) -- +(135.000000:14.142136pt and 29.698485pt) -- +(225.000000:14.142136pt and 29.698485pt) -- cycle;
\draw (40.000000, 15.000000) node {{V}};
\end{scope}
% Line 7: a b c G \rotatebox{90}{$U\in\mathcal{U}$}
\draw (68.000000,30.000000) -- (68.000000,0.000000);
\begin{scope}
\draw[fill=lavenderblue] (68.000000, 15.000000) +(-45.000000:8.485281pt and 29.698485pt) -- +(45.000000:8.485281pt and 29.698485pt) -- +(135.000000:8.485281pt and 29.698485pt) -- +(225.000000:8.485281pt and 29.698485pt) -- cycle;
\clip (68.000000, 15.000000) +(-45.000000:8.485281pt and 29.698485pt) -- +(45.000000:8.485281pt and 29.698485pt) -- +(135.000000:8.485281pt and 29.698485pt) -- +(225.000000:8.485281pt and 29.698485pt) -- cycle;
\draw (68.000000, 15.000000) node {\rotatebox{90}{$U\in\mathcal{U}$}};
\end{scope}
% Done with gates; drawing ending labels
% Done with ending labels; drawing cut lines and comments
% Done with comments
\end{tikzpicture}
\subcaption{A process tomography scheme without ancilla qubits. Random rotations are applied in the beginning and end of the circuit evaluation.}
\end{subfigure}
\hfill
\begin{subfigure}{0.4\linewidth}
%! \usetikzlibrary{decorations.pathreplacing,decorations.pathmorphing}
\begin{tikzpicture}[scale=1.300000,x=1pt,y=1pt]
\filldraw[color=white] (0.000000, -7.500000) rectangle (104.000000, 67.500000);
% Drawing wires
% Line 1: a b W Q_i<
\draw[color=black] (0.000000,60.000000) -- (104.000000,60.000000);
%   Deferring wire label at (0.000000,60.000000)
% Line 1: a b W Q_i<
\draw[color=black] (0.000000,45.000000) -- (104.000000,45.000000);
\filldraw[color=white,fill=white] (0.000000,41.250000) rectangle (-4.000000,63.750000);
\draw[decorate,decoration={brace,amplitude = 4.000000pt},very thick] (0.000000,41.250000) -- (0.000000,63.750000);
\draw[color=black] (-4.000000,52.500000) node[left] {$Q_i$};
% Line 2: c W C_i<
\draw[color=black] (0.000000,30.000000) -- (104.000000,30.000000);
\filldraw[color=white,fill=white] (0.000000,26.250000) rectangle (-4.000000,33.750000);
\draw[decorate,decoration={brace,amplitude = 1.875000pt},very thick] (0.000000,26.250000) -- (0.000000,33.750000);
\draw[color=black] (-4.000000,30.000000) node[left] {$C_i$};
% Line 3: d e W
\draw[color=black] (0.000000,15.000000) -- (104.000000,15.000000);
\draw[color=black] (0.000000,15.000000) node[left] {$\ket 0$};
% Line 3: d e W
\draw[color=black] (0.000000,0.000000) -- (104.000000,0.000000);
\draw[color=black] (0.000000,0.000000) node[left] {$\ket 0$};
% Done with wires; drawing gates
% Line 5: d G H
\begin{scope}
\draw[fill=lavenderblue] (12.000000, 15.000000) +(-45.000000:8.485281pt and 8.485281pt) -- +(45.000000:8.485281pt and 8.485281pt) -- +(135.000000:8.485281pt and 8.485281pt) -- +(225.000000:8.485281pt and 8.485281pt) -- cycle;
\clip (12.000000, 15.000000) +(-45.000000:8.485281pt and 8.485281pt) -- +(45.000000:8.485281pt and 8.485281pt) -- +(135.000000:8.485281pt and 8.485281pt) -- +(225.000000:8.485281pt and 8.485281pt) -- cycle;
\draw (12.000000, 15.000000) node {H};
\end{scope}
% Line 6: e G H
\begin{scope}
\draw[fill=lavenderblue] (12.000000, -0.000000) +(-45.000000:8.485281pt and 8.485281pt) -- +(45.000000:8.485281pt and 8.485281pt) -- +(135.000000:8.485281pt and 8.485281pt) -- +(225.000000:8.485281pt and 8.485281pt) -- cycle;
\clip (12.000000, -0.000000) +(-45.000000:8.485281pt and 8.485281pt) -- +(45.000000:8.485281pt and 8.485281pt) -- +(135.000000:8.485281pt and 8.485281pt) -- +(225.000000:8.485281pt and 8.485281pt) -- cycle;
\draw (12.000000, -0.000000) node {H};
\end{scope}
% Line 7: a C d
\draw (33.000000,60.000000) -- (33.000000,15.000000);
\begin{scope}
\draw[fill=lavenderblue] (33.000000, 60.000000) circle(3.000000pt);
\clip (33.000000, 60.000000) circle(3.000000pt);
\draw (30.000000, 60.000000) -- (36.000000, 60.000000);
\draw (33.000000, 57.000000) -- (33.000000, 63.000000);
\end{scope}
\filldraw (33.000000, 15.000000) circle(1.500000pt);
% Line 8: b C e
\draw (39.000000,45.000000) -- (39.000000,0.000000);
\begin{scope}
\draw[fill=lavenderblue] (39.000000, 45.000000) circle(3.000000pt);
\clip (39.000000, 45.000000) circle(3.000000pt);
\draw (36.000000, 45.000000) -- (42.000000, 45.000000);
\draw (39.000000, 42.000000) -- (39.000000, 48.000000);
\end{scope}
\filldraw (39.000000, 0.000000) circle(1.500000pt);
% Line 9: a b c G:width=20 \rotatebox{90}{V}
\draw (64.000000,60.000000) -- (64.000000,30.000000);
\begin{scope}
\draw[fill=lavenderblue] (64.000000, 45.000000) +(-45.000000:14.142136pt and 29.698485pt) -- +(45.000000:14.142136pt and 29.698485pt) -- +(135.000000:14.142136pt and 29.698485pt) -- +(225.000000:14.142136pt and 29.698485pt) -- cycle;
\clip (64.000000, 45.000000) +(-45.000000:14.142136pt and 29.698485pt) -- +(45.000000:14.142136pt and 29.698485pt) -- +(135.000000:14.142136pt and 29.698485pt) -- +(225.000000:14.142136pt and 29.698485pt) -- cycle;
\draw (64.000000, 45.000000) node {{V}};
\end{scope}
% Line 10: a b c d e G \rotatebox{90}{$U\in\mathcal{U}$}
\draw (92.000000,60.000000) -- (92.000000,0.000000);
\begin{scope}
\draw[fill=lavenderblue] (92.000000, 30.000000) +(-45.000000:8.485281pt and 50.911688pt) -- +(45.000000:8.485281pt and 50.911688pt) -- +(135.000000:8.485281pt and 50.911688pt) -- +(225.000000:8.485281pt and 50.911688pt) -- cycle;
\clip (92.000000, 30.000000) +(-45.000000:8.485281pt and 50.911688pt) -- +(45.000000:8.485281pt and 50.911688pt) -- +(135.000000:8.485281pt and 50.911688pt) -- +(225.000000:8.485281pt and 50.911688pt) -- cycle;
\draw (92.000000, 30.000000) node {\rotatebox{90}{$U\in\mathcal{U}$}};
\end{scope}
% Done with gates; drawing ending labels
% Done with ending labels; drawing cut lines and comments
% Done with comments
\end{tikzpicture}
\subcaption{Process tomography scheme using ancilla qubits. The ancilla qubits are entangled with the input of the channel, hence, applying rotations on ancillas become equivalent to random initializations.}
\end{subfigure}
\hfill~
\caption{Two equivalent circuits for shadow process tomography of {quantum channels defined by a circuit fragment containing gate} $V$. {Depicted above, $\mathcal C_i$ represent the circuit inputs (which can be assumed to be $\ket 0$) and $\mathcal Q_i$ represent the quantum inputs}. Unitaries picked from a measurement primitive $\mathcal U$ can be \textbf{(left)} applied on both the quantum inputs and outputs directly, or \textbf{(right)} applied at the end of the circuit using ancillas.
%The quantum gate $V$ represents an arbitrary (composition of) quantum gate(s).
}
    \label{fig:shadowprocess}
\end{figure}

Just like in shadow state tomography, we pick a measurement primitive $\mathcal U$. 
Then, there are two equivalent ways of obtaining the classical shadow, as depicted in Figure \ref{fig:shadowprocess}. 
The first method straightforwardly applies random unitaries on both the input (immediately before the channel) and output prior to measuring the outcome.
Equivalently, one can accomplish the same task by having the same number of ancilla qubits prepared as the dimension of the input, maximally entangle the ancilla and input qubits, and apply the random unitary altogether to the $n + m$ qubits (output + ancilla) at the end prior to measurement. 
Our numerical implementation uses the latter strategy, and measuring gives bitstrings $\hat b \in \{0,1\}^{n+m}$ that can be used to build the shadow of the Choi matrix
\begin{align}
    \hat \Lambda = \mathcal M\inv \left( U^\dagger \ket{\hat b} \bra{\hat b} U \right).
\end{align}
All procedures and results from state tomography are now transferable to process tomography. 
That is, to query polynomially many input-output relationships, the sample complexity scales only logarithmically with the number of these relationships. 
As each fragment can be viewed as a quantum channel from its quantum inputs to all of its outputs, the capability of building classical shadows of quantum channels becomes crucial for obtaining classical shadows of fragmented circuits.

\section{Classical shadows of fragmented circuits} \label{sec:divide-and-conquer}

A divide and conquer algorithm for classical shadow tomography follows directly from combining the results of the two previous sections. 
Eq.~\eqref{eq:gencutchoi} writes the expectation of an observable $O$ with respect to a circuit-induced state $\rho$ as sums and products of expectations with respect to Choi states of necessary fragments.
Thus, to obtain classical shadows of a fragmented circuit, one simply samples shadows from the Choi state via the channel-state duality and shadow process tomography. 
Like any circuit-cutting procedure, shadow tomography can be performed independently on each fragment in a distributed manner.
Moreover, this method does not assume the observable of interest, $O$, a priori. 
Having classical shadows of all fragments is equivalent to having classical shadows of the uncut state prepared by a circuit.
Subsequent circuit reductions, as instructed in Eq.~\eqref{eq:reduction}, can be performed once an observable $O$ is chosen.
Lastly, classical shadows of fragments can be reused when multiple observables, $O_1, O_2, \dots, O_M$, are of interest. 
As in the case of ordinary shadow tomography, the number of measurements needed grows logarithmically with $M$.

\subsection{Sample complexity of divide and conquer}\label{sec:sft-math}
Once the shadows of fragments are obtained, it is preferable to have theoretical guarantees that once operations in Eq.~\eqref{eq:gencutchoi} are performed, the final estimate of $\tr(O\rho)$ is close to the true value. 
Below, we present the sample complexity of each fragment for estimating the expectation of a single observable $O$. 
For each fragment, the same ensemble of classical shadows will be used to obtain information about the quantum input and output in addition to the circuit output if present.
As a result, the sample complexity of each fragment is highly dependent on the degree of the fragment on the corresponding multigraph. 
{We've compiled a list of notation introduced in prior section in Table \ref{tab:notation}, and } the formal statement is provided below.

\begin{table}[]
    \centering
    {\begin{tabular}{|c|c|}
        \hline
        \textit{Symbol} & \textit{Definition} \\
        \hline
        $N_q$ & the number of qubits in the quantum circuit \\
        \hline
        $O$ & the observable we hope to estimate the expectation of \\
        \hline
        $\mathcal K$ & the set of qubits that $O$ acts non-trivially on \\
        \hline 
        $G = (F,E)$ & a directed graph denoting how circuit fragments relate to each other; \\ & $F$ is the set of fragments and $E$ is the set of edges each corresponding to a cut \\
        \hline
        $f \in F$ & a fragment of the cut circuit \\
        \hline
        $\mathcal Q_i(f)$ & the ``quantum inputs'' of fragment $f$ (the in-degrees of the fragment graph) \\
        \hline 
        $\mathcal Q_o(f)$ & the ``quantum outputs'' of fragment $f$ (the out-degrees of the fragment graph) \\
        \hline 
        $\mathcal C_o(f)$ & the ``circuit output'' of the fragment (the qubits of $f$ being acted on by $O$) \\
        \hline
    \end{tabular}}
    \caption{A table of used notation in Theorem \ref{thm:mainthm}.}
    \label{tab:notation}
\end{table}

\begin{theorem}\label{thm:mainthm}
Consider a quantum state $\rho$ from a circuit that is cut into fragments with the corresponding reduced fragment graph $G = (F, E)$ (see Section \ref{sec:gencirccut}) and a $|\mathcal K|$-local observable $O$ {that factorizes across the given fragmentation scheme}. \red{For any fragment $f \in F$}, define the \emph{degree} of a fragment to be the sum of its number of quantum inputs, quantum outputs, and circuit outputs 
\begin{align}
    \red{\deg(f) = Q_i(f) + Q_o(f) + C_o(f),}
    &&
    \red{\qdeg(f) = Q_i(f) + Q_o(f)}.
\end{align}
Then, for any $\epsilon, \delta \in (0,1)$ and a fixed fragment \red{$f$}, let 
\begin{align}
    \red{K_{f} = 2 \log \frac{ 2 |F| \cdot 4^{\qdeg(f)} }{\delta},}
    &&
    \red{N_{f} = \frac{34 |F|}{\epsilon^2} 4^{\deg(f)+2|E|} \|O\|_{\text{S}}^4,}
    \label{eq:mainthmeq}
\end{align}
where $\|\cdot \|_S$ is the Pauli-string norm, \textit{i.e.,} $\|O\|_{\text{S}} = \sqrt{\tr(O^\dagger O)/2^{N_q}}$. Then, a collection of \red{$N_{f} \cdot K_{f}$} independent classical shadows for each fragment \red{$f \in F$}  can estimate $\tr(O\rho)$ with (expected) additive error $\epsilon$ with probability at least $1 - \delta$ (neglecting higher order terms).
\end{theorem}

The proof is deferred to Appendix \ref{sec:main-thm-pf}.
\red{Theorem \ref{thm:mainthm}} gives the ``per-fragment'' sample complexity for estimating the expectation of an observable with respect to an uncut circuit.
\red{
We keep the factors of $K_f$ and $N_f$ distinct for consistency and ease of cross-referencing with Lemma \ref{thm:shadowtomo} (Theorem S1 of Ref.~\cite{huang2020predicting}).
These factors can be traced to a median-of-means estimation procedure, where by an estimator is obtained by taking the median of $K$ independent sample means of size $N$.
}
Similar to the existing analysis of circuit cutting in the literature, the cost of estimating expectations using fragmentation scales exponentially with the number of cuts. 

{
\begin{remark}
For notational cleanliness, Theorem \ref{thm:mainthm} states the sample complexity for estimating the expectation of a single observable.
However, the benefits of performing shadow tomography, namely, the ability to estimate multiple observables using the same ensemble of classical shadows, still remains. In the case of estimating a single observable, we require no a priori knowledge of the observable aside from locality on each fragment. On the other hand, if we want to estimate the expectation of $M$ observables that factorize across the same circuit cutting scheme, the per-fragment sample complexity becomes \red{$K_f \cdot N_f$, where}
\begin{align}
    \red{K_{f} = 2 \log \frac{ 2 |F| \cdot 4^{\qdeg(f)} M }{\delta},}
    &&
    \red{N_{f} = \frac{34 |F|}{\epsilon^2} 4^{\max_{i \in [M]} \deg(f; O_i)+2|E|}  \max_{i \in [M]} \|O_i\|_{\text{S}}^4,}
\end{align}
\red{and $\deg(f; O_i)$} refers to the degree of the fragment when the desired observable is $O_i$. As one would expect, the sample complexity grows logarithmically with respect to $M$ barring other parameters of the model.
\end{remark}
}

{
\begin{remark}
The assumption that the observable factorizes across the fragments in Theorem~\ref{thm:mainthm} can also be relaxed. As analyzed in Section \ref{sec:gencirccut}, one can always decompose an observable into a linear combination of {factorized observables} (e.g., Pauli strings) and repeat the divide-and-conquer procedure on each element of the decomposition. The analysis is omitted in Theorem~\ref{thm:mainthm} as decomposing an arbitrary quantum observable can introduce exponentially many terms, further convoluting the existing bound. We note that, in practice, observables often admit a suitable decomposition involving only polynomially many basis elements, in which case, the sample complexity will also scale polynomially.
\end{remark}
}

The comparison of the divide-and-conquer scheme with regular shadow tomography without cuts heuristically boils down to contrasting $\sum_{f \in F} 4^{\deg(f_k) + 2|E|}$ with $4^{\size(O)}$. For circuit cutting to be effective, we should expect the size of the observable to be sufficiently large to justify paying the overhead associated with inferring degrees of freedom near the cut without excessive fragmentation. 
\red{We carry out some further calculations in the remark below.}

\red{
\begin{remark}\label{rmk:heurstics}
Consider observable $O$ with $\|O\|_\infty = \|O\|_\text{S} \approx 1$ (e.g., because $O$ is a Pauli string). Then, circuit cutting is beneficial when
\begin{align}
    \sum_{f \in F} K_f \cdot N_f \sim \sum_{f \in F} \log\left(\frac{|F| 4^{\qdeg(f)}}{\delta}\right) \frac{|F|}{\epsilon^2} 4^{\deg(f) + 2|E|}
    \lesssim \frac{\log \delta \inv}{\epsilon^2} 4^{\size(O)},
\end{align}
where we can bound the sum over $f\in F$ by $|F|$ times the maximal summand, arriving at the inequality
\begin{align}
  C_o^{\mathrm{max}}(O) + \textnormal{max-qdeg}(F) + 2|E| + 2\log|F| + \log\left(1 + \frac{\log|F| + \textnormal{max-qdeg}(F)}{\log\delta^{-1}}\right) \lesssim \size(O).
  \label{eq:advantage_full}
\end{align}
Here $C_o^{\mathrm{max}}(O)=\max_f C_o(f)$ is the maximum number of qubits addressed by $O$ in any fragment, and $\textnormal{max-qdeg}(F) = \max_{f \in F} \qdeg(f)$ is the maximum number of cuts incident to a fragment.
It is generally the case that $\log|F|$ is negligible compared to $|E|$, and as long as $\log|F| + \textnormal{max-qdeg}(F)$ is not exponentially large in some relevant parameter---in which case Eq.~\eqref{eq:advantage_full} cannot be satisfied anyways---the final term on the left-hand side of Eq.~\eqref{eq:advantage_full} will be small.
Altogether, the (approximate) condition for an asymptotic advantage from circuit cutting becomes
\begin{align}
  C_o^{\mathrm{max}}(O) + \textnormal{max-qdeg}(F) + 2|E| \lesssim \size(O).
  \label{eq:advantage_full_simplified}
\end{align}
In the extreme case when $\size(O) = 1$, this inequality is never achieved because circuit cutting adds unfavorable degrees of freedom to the fragments.
As $\size(O)$ increases, however, one can eventually expect an advantage from circuit cutting.
%However, on the other extreme when $\size(O) = N_q$, then there may be an advantage to circuit cutting; we will give explicit examples in the coming paragraphs.
Note that only the subgraph in the past light cone of $O$ needs to be considered in Eq.~\eqref{eq:advantage_full_simplified}.
\end{remark}
}

\red{From the remark above, we see that Theorem~\ref{thm:mainthm} is difficult to interpret as-is due to its generality.} For an $N_q$-qubit circuit, we therefore introduce the cascade and cluster circuit ansatz, both illustrated in Figure~\ref{fig:clustered-ansatz}. By making concrete assumptions on circuit structure, we present variants of Theorem~\ref{thm:mainthm} tailored to these two circuit ansatzes. Via these two circuit structures, we will \red{analytically and} numerically study the trade-off between the size of the observable and the level of fragmentation in the next subsection (see Section~\ref{sec:numerics}).
% \blue{Reviewer one asked for regimes where divide and conquer has a lower sample complexity than regular divide and conquer, which from experience, require an unrealistic number of shots. Should we a) provide the figure in the main text, b) provide the figure in the appendix, or c) defer the reviewer's comment?}

\begin{figure}[t]
    ~\hfill
    \begin{subfigure}{.45\linewidth}
        \includegraphics[width=\linewidth]{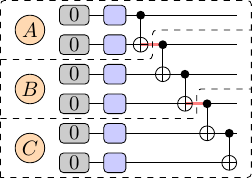}
        \caption{Cascade circuit ansatz with three fragments.}
        \label{fig:cascade}
    \end{subfigure}
    \hfill
    \begin{subfigure}{.45\linewidth}
        \includegraphics[width=\linewidth]{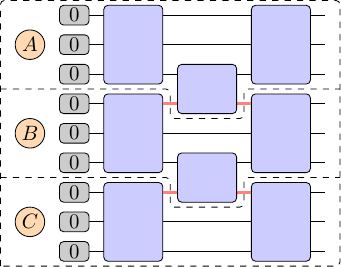}
        \caption{Clustered circuit ansatz with three fragments, reproduced from Ref.~\cite{perlin2021quantum} with permission from the authors.}
        \label{fig:clustered}
    \end{subfigure}
    \hfill~
    \caption{Circuit ansatze used in experiments. Each purple block represents a Haar-random gate on the corresponding qubits. Red lines indicate the location of cuts, and individual fragments are labeled by letters and separated by dashed lines.}
    \label{fig:clustered-ansatz}
\end{figure}

{
\paragraph{Cascade Circuit Ansatz} The cascade circuit ansatz (Figure~\ref{fig:cascade}) features a GHZ-like circuit structure. Performing $K$ cuts (with $K < N_q$) produces $K+1$ fragments, and we will label each fragment in the order with the qubit index. Each fragment has a single quantum input (except for the first fragment) and a single quantum output (except for the last fragment). Notably, if the observable is local to a subset of qubits $\mathcal K$, all fragments after qubit $\max \mathcal K$ can be ignored. Combining this information, we present the Theorem \ref{thm:mainthm} assuming the cascade circuit ansatz.
}

{
\begin{corollary}\label{thm:cascade-corollary}
Consider an $N_q$-qubit quantum state $\rho$ from a cascade circuit (see Section \ref{sec:gencirccut}) amendable to $K$ cuts, and an observable $O$ local to qubits $\mathcal K \subset [N_q]$. Then, for any $\epsilon, \delta \in (0,1)$ and a fixed fragment \red{$f \in F$}, let 
\begin{align}
    \red{K_{f} = 2 \log \frac{32|F|}{\delta},}
    &&
    \red{N_{f} = \frac{34 |F|}{\epsilon^2} 4^{\mathcal C_o(f) + 2|F|} \|O\|_{\text{S}}^4,}
    \label{eq:cascade-corollary}
\end{align}
where $|F| = K+1$.
Then, a collection of \red{$N_{f} \cdot K_{f}$} independent classical shadows for each fragment \red{$f \in F$}  can estimate $\tr(O\rho)$ with (expected) additive error $\epsilon$ with probability at least $1 - \delta$ (neglecting higher order terms).
\end{corollary}
}
\red{Carrying out calculations similar to that of Remark \ref{rmk:heurstics} gives that divide-and-conquer reduces sample complexity when
% for the particular case of the cascade circuit ansatz. Suppose again that $\|O\|_\infty = \|O\|_\text{S} = 1$. 
% In this case, circuit cutting is asymptotically advantageous when
% \begin{align}
%     \sum_{f \in F} K_f \cdot N_f
%     \lesssim \log\left(\frac{|F|}{\delta}\right)
%     \frac{|F|^2}{\epsilon^2} \cdot
%     4^{\mathcal C_o^{\mathrm{max}}(O) + 2|F|}
%     \lesssim \frac{\log(1/\delta)}{\epsilon^2} \cdot 4^{\size(O)},
% \end{align}
% where we ignore constant factors for simplicity and $C_o^{\mathrm{max}}(O) = \max_{f \in F} \mathcal C_o(f)$.
% Taking the logarithm of both sides and assuming $\log(\delta^{-1})\gg\log|F|$, we thus find that the condition for an asymptotic advantage with a cascade circuit becomes
\begin{align}
  C_o^{\mathrm{max}}(O) + 2|F| \lesssim \size(O).
\end{align}
For example, when the size of the observable matches the size of the circuit, i.e., $\size(O) = N_q$, then $C_o^{\mathrm{max}}(O) = N_q/|F|$ and circuit cutting is advantageous when there are at least $N_q/|F|\gtrsim2$ qubits per fragment.
}

{
\paragraph{Clustered Cascade Circuit Ansatz} The cascade circuit ansatz (Figure~\ref{fig:clustered}) features ``clusters'' of quantum gates that are staggered next to each other. There are clusters of two sizes, call them $c_1$ and $c_2$ with $c_1 > c_2$. Viewed horizontally, sizes of clusters alternate starting from a layer of $c_1$-qubit clusters followed by clusters of size $c_2$ placed between adjacent clusters of size $c_1$. Let $d$ be the depth of the circuit, and $\ell = N_q / c_1$ be the number of ``layers'' of clusters laid vertically. For example, in Figure~\ref{fig:clustered}, $c_1 = 3$, $c_2 = 2$, $d = 3$, and $\ell = 3$.
}

{The cuts will be performed right before and after the clusters of size $c_2$, splitting the circuit into $\ell$ fragments. Each fragment then has $(d-1)/2$ quantum inputs and $(d-1)/2$ quantum outputs. Moreover, unlike the cascade circuit, the resulting fragment graph is cyclic and no circuit reduction can be performed. With this in mind, we can calculate the sample complexity of our divide-and-conquer scheme assuming the clustered circuit ansatz.}

{
\begin{corollary}\label{thm:clustered-corollary}
Consider an $N_q$-qubit quantum state $\rho$ from a clustered circuit with parameters $(c_1, c_2, d, \ell)$, and an observable $O$ local to qubits $\mathcal K \subset [N_q]$. Then, for any $\epsilon, \delta \in (0,1)$ and a fixed fragment \red{$f \in F$}, let 
\begin{align}
    \red{K_{f} = 2 \log \frac{ 2 \ell \cdot 4^{\qdeg(f)} }{\delta},}
    &&
    \red{N_{f} = \frac{34 \ell}{\epsilon^2} 4^{C_o(f)+2(\ell-1)(d-1)} \|O\|_{\text{S}}^4.}
    \label{eq:clustered-corollary}
\end{align}
Then, a collection of \red{$N_{f} \cdot K_{f}$} independent classical shadows for each fragment \red{$f \in F$}  can estimate $\tr(O\rho)$ with (expected) additive error $\epsilon$ with probability at least $1 - \delta$ (neglecting higher order terms).
\red{Here $\qdeg(f) \le 2(d-1).$}
\end{corollary}
}
\red{Again using Remark \ref{rmk:heurstics} we find that circuit cutting for clustered circuits with the structure in Fig.~\ref{fig:clustered} is asymptotically advantageous when
% with the cascade circuit ansatz. Using Corollary \ref{thm:cascade-corollary}, we are interested in the inequality:
% \begin{align}
%     \sum_{f \in F} K_f \cdot N_f \sim \log \frac{\ell 4^{d-1}}{\delta} \cdot \frac{\ell^2}{\epsilon^2} 4^{C_o^{\mathrm{max}}(O) + (2\ell+1)(d-1)} \lesssim \frac{\log \delta \inv}{\epsilon^2} 4^{\size(O)}.
% \end{align}
% Now, we will assume that $\log \ell \ll \log \delta \inv$ so that
% \begin{align}
%     \log \log \frac{\ell 4^{d-1}}{\delta} \sim 2 \log d \cdot \log \log \frac{1}{\delta},
% \end{align}
% then taking logarithms on both sides gives the condition that
\begin{align}
    C_o^{\mathrm{max}}(O) + 2\ell(d-1) \lesssim \size(O),
\end{align}
where we note that here $\textnormal{max-qdeg}(F)=2(d-1)$, $2|E|=2(\ell-1)(d-1)$, and $\ell=|F|$.
}

\subsection{Numerical Results}
\label{sec:numerics}

Many components involved in Theorem \ref{thm:mainthm} are dependent on the particular circuit that one chooses to run. Hence, we would like to establish some preliminary numerical results that depict the differences between shadow tomography on unfragmented circuits and our proposed divide and conquer method. We implement the algorithm by adapting Pennylane's circuit cutting and shadow tomography functions \cite{bergholm2018pennylane}. We run our experiments either using a clustered-circuit ansatz or a cascade-circuit ansatz (see Figure \ref{fig:clustered-ansatz}) comprised of Haar-random gates. Each circuit was paired with a randomly generated Pauli observable of a chosen size. We examined the scaling behavior of the two methods with respect to the number of samples (denoted by shots), the size of the observable (denoted by $\size(O)$), and the number of fragments (denoted by $|F|$). We repeat each experimental setting 250 times with different randomly-generated circuit-observable pairs. The results are displayed in Figure \ref{fig:main-num}.

\begin{figure}
    \centering
    \includegraphics[width=.6\linewidth]{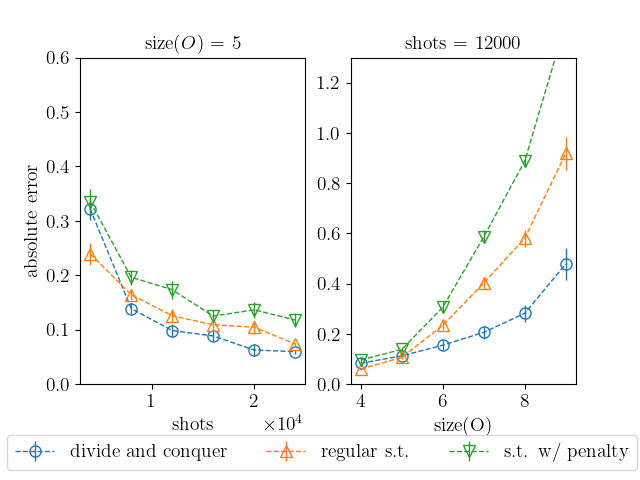}
    \caption{Numerical results comparing estimates obtained from shadow tomography with (blue) and without (orange) fragmentation using the clustered-circuit ansatz (see Figure \ref{fig:clustered}) {with $12$ qubits and $|F| = 3$ fragments}. The estimate can be penalized (green) if it was an uninformed guess (clarified in main text). The lower the absolute error, the closer the estimated expectation is to the true expected value. \textbf{(left)} The absolute error with respect to the number of shots (classical shadows) across different numbers of fragments. \textbf{(right)} The absolute error with respect to the size of the randomly generated observable across different numbers of fragments.}
    \label{fig:main-num}
\end{figure}

\paragraph{Advantageous regimes} 
For a fixed circuit, fragmentation significantly improves the estimated expectation when large observables are being estimated{, which is} evident in {the right panel of Figure \ref{fig:main-num}}. As the weight of an observable $O$ grows, the difference in accuracy between regular shadow tomography ($|F|=1$) and fragmented shadow tomography increases. Since each fragment is only ``responsible'' for the part of $O$ that directly affects the qubits in the fragment, circuit cutting divides the observable into smaller parts. Recall that the sample complexity of shadow tomography scales exponentially with respect to the weight of the observable. As the number of fragments increases, the average per-fragment weight of each observable decreases proportionally, which then lowers the sample complexity by an exponential factor. This shrinkage in the size of $O$ on a fragment results in an increasing separation between fragmented and unfragmented shadow tomography. 

\begin{figure}
    \centering
    \includegraphics[width=.6\linewidth]{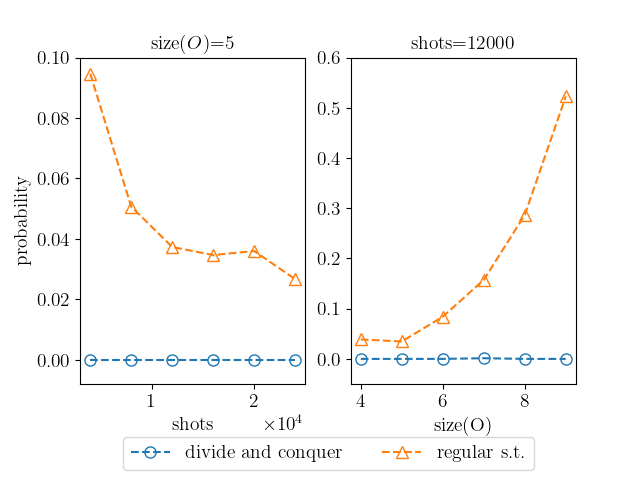}
    \caption{The probability of not observing the Pauli string of interest using (blue) fragmented and (orange) ordinary shadow tomography in Figure \ref{fig:main-num}, shown with respect to \textbf{(left)} the number of shots and \textbf{(right)} the size of the Pauli string.}
    \label{fig:unobs-seq}
\end{figure}

However, the justification above significantly understates the benefit of fragmentation when estimating expectations of large observables. Recall Section \ref{sec:st-code}---if the sequence of the observable being estimated is not present in the ensemble of classical shadows, we return a default value of zero. Meanwhile, for any Pauli observable $P$ and state $\rho$ induced by the circuit ansatz in Figure \ref{fig:clustered-ansatz} comprised of only Haar-random unitaries, one should expect that $\E_{(\rho,P) \sim \text{Ansatz}} \tr(P\rho) 
% = \tr(P/2^n) 
= 0$, where $\E_{(\rho,P) \sim \text{Ansatz}}$ denotes an average over choices of Haar-random gates defining $\rho$ and fixed-size observables $P$.
Moreover, we can also expect a shrinking variance with $(\rho,P) \sim \text{Ansatz}$ in the limit of large circuit size, which means that the individual terms in this average approach zero as well.
Although under regular circumstances, outputting the uniform prior (zero) as an a priori estimate is the least-informed guess that one can make, in our numerical experiment, zero is actually a high quality guess!

We can study the empirical probability of such phenomenon occurring, see Figure \ref{fig:unobs-seq}. It turns out that for large observables, the desired Pauli string is only observed roughly half of the times when performing shadow tomography with the clustered-circuit ansatz and 12000 shots. The number of shots is an influencing factor, too. However, it has much less of an effect compared to the observable weight. Nonetheless, both trends can be explained through the same rationale: the number of shots dedicated to each degree of freedom is too low. Meanwhile, fragmentation negates this problem completely, as seen by the blue line in Figure \ref{fig:unobs-seq}. To further emphasize this point, we provide a (moderate but arbitrary) penalty of 1 for making an uninformed guess for the expectation value of an operator, shown in the green lines in Figure \ref{fig:main-num}. The penalty widens the already existing gap between regular and fragmented shadow tomography, especially when the number of shots is low and the weight of the observable is high. In conclusion, we expect fragmentation to be a beneficial procedure when the observable is large.

\begin{figure}
    \centering
    \includegraphics[width=.6\linewidth]{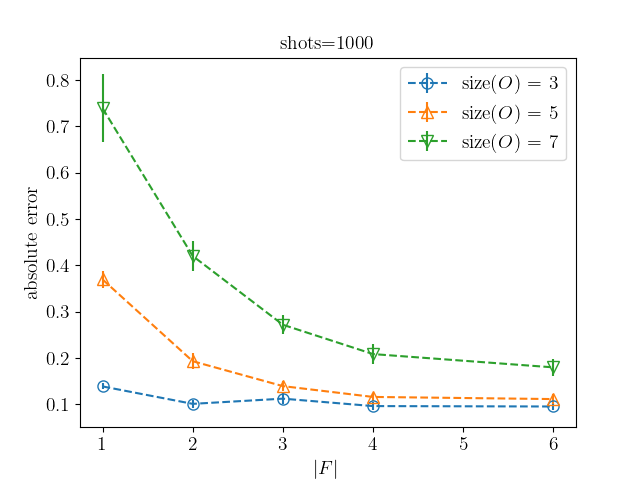}
    \caption{Numerical results showing the error scaling with respect to the number of fragments using the cascade circuit ansatz (see Figure \ref{fig:cascade}). Empirical results show that error decreases as the number of fragment increases, more dramatically so for large observables.}
    \label{fig:fragment-scaling}
\end{figure}

Moreover, absolute error also decreases with the number of fragments given a fixed total shot count. Figure \ref{fig:fragment-scaling} shows numerical scaling of error with the number of fragments in the cascade circuit ansatz. We see that, even under relatively low shot counts, absolute error empirically vanishes with the fragment size. This goes to show that, given sufficient number of shots, the overhead incurred by characterizing quantum degrees of freedom would not significantly interfere with the quality of estimation using classical shadows.

\paragraph{Limits of fragmentation}
Though there are circumstances where fragmentation is beneficial, is it necessarily the case that more fragments the merrier? The answer is no! First, it is evident from the top row of Figure \ref{fig:main-num} that fragmentation is not advantageous when low number of shots are available. Recall that there is an extra degree of freedom to infer per incident of cut on a fragment. Consequently, the number of samples needed to properly account for the additional uncertainty grows exponentially with the number of incident cuts on a fragment \red{(see also Remark \ref{rmk:heurstics})}. Thus, when quantum resources are scarce, the drawback of not having enough shots to characterize each fragment well often outweighs the benefit of avoiding having unobserved Pauli strings as described before, resulting in poor estimates.

Beside the point above, the limitation is not well understood. Theorem \ref{thm:mainthm} predicts an eventual blow-up of error as the number of fragments increases while fixing the total shot count. Though some qualitative behavior from theory can be replicated empirically by dramatically lowering the total number of shots, theoretical predictions are extremely pessimistic on the propagation of error during tensor contraction, and deviates far from empirical numerical results. We present preliminary evidence for this in Appendix \ref{sec:opt-fragment}, and we defer the development of a tighter, more quantitatively informative bound to future work.

\section{Conclusions} \label{sec:conclusion}

In this paper, we introduced fragmented shadow tomography, a novel divide-and-conquer method for building classical shadows while incorporating circuit cutting procedures. This method combines the benefits of both of its predecessors---efficient computation and storage of quantum states through classical shadows and the ability to be parallelized in independent quantum simulators/devices. We derived a general formula to combine shadow tomography estimates of each fragment. Moreover, we provide the respective sample complexity for each fragment such that the total additive error is bounded to desired precision. Lastly, we numerically demonstrated its advantages over traditional, unfragmented shadow tomography, exceedingly so when large observables are of interest. 

We focused on building classical shadows including all components---quantum input, quantum output, and circuit output---of a fragment. This need not be the case. One can also consider partially derandomizing fragmented shadow tomography by building classical shadows only on circuit outputs. That is, for all quantum inputs and outputs, \textit{i.e.,}~$\bm M \in \mathcal B^{|E_{\not \Delta}|}$, we can replace randomized Pauli measurements with deterministic initialization/measurement schemes, keeping the randomized procedure only for circuit outputs. Each classical shadow will be a function of $M_{f_k}$ and the respective Choi matrices can be inferred via maximum likelihood procedures (see \cite{perlin2021quantum}). The two methods are equivalent in the large sample limit. However, we expect derandomization will bring more statistical stability under finite computational resources.
Derandomization may also reduce errors, as in Ref.~\cite{huang2020predicting}.

\red{Both analytically and in numerical experiments}, we see a trade-off between the number of fragments and the size of the observable. For small observables, excessive fragmentation leads to accumulation of error and decrease in number of shots attributed to each fragment. On the other hand, dividing circuits becomes beneficial for large observables that would require excessively many shots to estimate. Thus, we suspect an existence of an optimal number of cuts given a quantum circuit and an observable (or a priori knowledge on the size of the observable). See Appendix \ref{sec:opt-fragment} for preliminary evidence. As a consequence of this conjecture, predicting the optimal number of fragments subsequently becomes an important open question to further leverage quantum divide-and-conquer algorithms.

\paragraph{Acknowledgement}
The authors thank Nathan Killoran and Thomas Bromley for helpful discussions.
M.A.P. thanks Pranav Gokhale for enabling an environment in which this work was possible.
Z.H.S is supported by the
U.S. Department of Energy Office of Science National
Quantum Information Science Research Center, QNEXT. D.C. was funded by the SULI fellowship at Argonne National Laboratory. 
 
The submitted manuscript has been created by UChicago Argonne, LLC, Operator of Argonne National Laboratory (``Argonne”). Argonne, a U.S. Department of Energy Office of Science laboratory, is operated under Contract No. DE-AC02-06CH11357. The U.S. Government retains for itself, and others acting on its behalf, a paid-up nonexclusive, irrevocable worldwide license in said article to reproduce, prepare derivative works, distribute copies to the public, and perform publicly and display publicly, by or on behalf of the Government. The Department of Energy will provide public access to these results of federally sponsored research in accordance with the DOE Public Access Plan (\url{http://energy.gov/downloads/doe-public-access-plan}).

\paragraph{Conflict of Interest} The authors declare that they have no conflict of interest.

\printbibliography

\appendix
\section{Proof of Theorem \ref{thm:mainthm}}\label{sec:main-thm-pf}

The error analysis for recombining estimators from fragments requires simple facts regarding propagation of error when taking sums and products of perturbed quantities. 
These tools are proved as lemmas below.

\begin{lemma}\label{thm:errorlemma}
Let $(\hat a_i)_{i=1}^n$ be a collection of independent random variables satisfying $\E (\hat a_i - a_i) = 0$, $|a_i| \leq 1$, and {$\E (\hat a_i  - a_i)^2 \leq \epsilon^2$ for some $\epsilon > 0$}. Then, 
{
\begin{align}
    \std \left( \prod_i \hat a_i\right) \leq \sqrt{n} \epsilon + r(n,\epsilon)
\end{align}
where $r(n,\epsilon) = \left( \sum_{k=2}^n {n\choose k} \epsilon^{2k} \right)^{1/2}$}.
\end{lemma}

\begin{proof}
{
Using the independence of $\hat a_i$ and that they are less than 1 in magnitude, we get
\begin{align}
    \E \left(\prod_i \hat a_i - \prod_j a_j\right)^2
    = \prod_i (\epsilon^2 + a_i^2) - \prod_j a_j^2 
    = \epsilon^2 \sum_i \prod_{j\ne i} a_j^2 + \sum_{k=2}^n \sum_{S \subseteq [n]} \epsilon^{2k} \prod_{i \in S} a_i^2 \leq n \epsilon^2 + r(n,\epsilon).
\end{align}}
\end{proof}

{
\begin{remark}
Note that $r(n,\epsilon)$ is exponential with respect to $n$. However, when applying it later in the proof, we will consider a variance of $\epsilon/\sqrt{n}$ rather than $\epsilon$. In that case, the remainder term is well-controlled by $\epsilon^2$:
\begin{align}
    r\left(n, \frac{\epsilon}{\sqrt{n}} \right)^2 = \sum_{k=2}^n {n \choose k} \left( \frac{\epsilon^2}{n} \right)^k \leq \sum_{k=2}^n \frac{\epsilon^{2k}}{k!} \leq \epsilon^{4}.
\end{align}
\end{remark}
}

{In addition, we review the \textit{Schwarz Inequality} \cite{ccinlar2011probability} in the context of probability theory.
\begin{lemma}\label{thm:cauchy-schwartz}
    Let $\hat a$ and $\hat b$ be two random variables with means $a$ and $b$ respectively. Then the following inequality holds: 
    \begin{align}
        \left \vert \Cov(\hat a, \hat b) \right \vert \leq \std(\hat a) \std(\hat b)
    \end{align}
    where $\Cov(\hat a, \hat b) = \E(\hat a - a)(\hat b - b)$. 
\end{lemma}
}

\begin{proof}[Proof of Theorem \ref{thm:mainthm}]
Without loss of generality, assume that all fragments are in the past light cone of the operator $O$ (see Section \ref{sec:gencirccut}). Recall the circuit cutting formula presented in Eq.~\eqref{eq:gencutchoi}.  Using procedures from shadow tomography (see Lemma \ref{thm:shadowtomo}), we seek to estimate each trace-term in the product to $\eta$ additive error with probability $1 - \gamma$. Our goal is to find the corresponding $\eta$ and $\gamma$ such that the total error, upon adding and multiplying, gives an error of $\epsilon$ with probability $1 - \delta$. 

Let $\mu = \tr(O\rho)$ and $\hat \mu$ denote the estimator derived from the divide-and-conquer shadow tomography. Moreover, let $\tau(\bm M, P)$ be a shorthand for the product of trace terms in Eq.~\eqref{eq:gencutchoi} and $\hat \tau(\bm M, P)$ be the corresponding estimator, \textit{i.e.,}
\begin{align}
    \hat \mu = \sum_{P \in \{I,X,Y,Z\}^{\tensor |\mathcal K|}} \alpha_P \sum_{\bm M \in \mathcal B^{|E_{\not \Delta}|}} \hat \tau(\bm M, P)
\end{align}
{where 
\begin{align}
    \hat \tau(\bm M, P) = \prod_{f_k \in \kappa} \tr \left(\left( {M_{f_k}^{Q_i}}^\intercal \tensor M_{f_k}^{Q_o} \tensor P_{\mathcal C_o(f_k)} \right) \Lambda_{f_k} \right) \cdot \prod_{f_i \in \Gamma} \tr \left( \left( {M_{f_i}^{Q_i}}^\intercal \tensor M_{f_i}^{Q_o} \tensor I_{\mathcal C_o(f_i)} \right) \Lambda_{f_i} \right).
\end{align}
}
We can bound the variance of $\hat \mu$: 
\begin{align}
    \Var(\hat \mu) &= \sum_{P,P'} \vert \alpha_P \vert^2 \vert \alpha_{P'} \vert^2 \,\Cov \left( \sum_{\bm M} \hat \tau \left(\bm M, P \right), \sum_{\bm M'} \hat \tau \left(\bm M', P' \right) \right) \\
    &\leq \sum_{P,P'} \vert \alpha_P \vert^2 \vert \alpha_{P'} \vert^2 \sum_{\bm M, \bm M'} \std\left( \hat \tau\left(\bm M, P\right) \right) \std \left((\hat \tau \left(\bm M', P'\right) \right) \\
    &= \left( \sum_{P, \bm M} \vert \alpha_P \vert^2 \std \left( \hat \tau \left( \bm M, P \right) \right) \right)^2 \\
    &\leq \left( \sum_{\bm M} \left( \sum_P |\alpha_P|^2 \right) \max_{P',\bm M'} \left\{ \std \left( \hat \tau \left( \bm M', P' \right) \right) \right\} \right)^2 \\
    &= \left( 4^{|E_{\not \Delta}|} \|O\|_{\text{S}}^2 \,\max_{P,\bm M} \left\{ \std \left( \hat \tau \left( \bm M, P \right) \right) \right\} \right)^2 \label{eq:pf-1}
\end{align}
where the first inequality is by Schwarz inequality (see Lemma \ref{thm:cauchy-schwartz}), and 
\begin{align}\label{eq:pauli-string-norm}
    \|O\|_{\text{S}} = \sqrt{\tr(O^\dagger O)/2^{N_q}}
\end{align}
is the ``Pauli-string'' norm of $O$.
If $O = \sum_P \alpha_P P$, where $\alpha_P$ are scalar coefficients for the Pauli strings $P$, then $\|O\|_{\text{S}}^2 = \sum_P \|\alpha_P\|^2$.

Note that, while $\hat \tau$'s are not independent of each other, terms inside $\hat \tau$ are; {that is, we can write $\hat \tau \left( \bm M, P \right) = \prod_i \hat{a}_i\left( \bm M, P \right)$ for random variables $\hat{a}_i\left( \bm M, P \right)$ each corresponding to one of the trace terms.} Thus, by Lemma \ref{thm:errorlemma}, we can bound the standard deviation of the last term in Eq.~\eqref{eq:pf-1}:
\begin{align}
    \max_{P,\bm M} \left\{ \std \left( \hat \tau \left( \bm M, P \right) \right) \right\} \leq {\sqrt{(|\kappa| + |\Gamma|)}}\eta + \mathcal O(\eta^2).
\end{align}
So, the standard deviation of $\hat \mu$ is bounded by
\begin{align}
    \std(\hat \mu) = 4^{|E_{\not \Delta}|}  \|O\|_{\text{S}}^2  \sqrt{\left(|\kappa| + |\Gamma| \right)}\, {\left(\eta + \mathcal O(\eta^2)\right)}.
\end{align}

Note that the upper bound above holds with probability at most $(1-\gamma)^{|F|} = 1 - |F|\gamma + \mathcal O(\gamma^2) = {1-\delta}$. Thus, using Lemma \ref{thm:shadowtomo}, we get the sample complexity per fragment by rescaling $\eta$ and $\gamma$ to be with respect to $\epsilon$ and $\delta$:
\begin{align}
    \red{K_{f} = 2 \log \frac{2|F|\cdot 4^{\qdeg(f)}}{\delta},}
    &&
    \red{N_f = \frac{34 \cdot 16^{|E_{\not \Delta}|} \left( |\kappa| + |\Gamma| \right) \|O\|_{\text{S}}^4}{\epsilon^2} 4^{\deg(f)}}.
\end{align}
{For each fragment, we must build shadows of the corresponding Choi states capable of inferring all quantum degrees of freedom, leading to the factor of \red{$4^{\qdeg(f)}$ in $K_{f}$} for estimating expectations for all \red{length-$\qdeg(f)$} Pauli strings ($M$ in Lemma \ref{thm:shadowtomo}); the $|F|$ factor comes from rescaling the success probability}. Moreover, the size of the operators acting on fragment \red{$f$} is the sum of the size of $O$ on the circuit output of $f_k$ in addition to the quantum degrees of freedom; hence, \red{$4^{\deg(f)}$} factor in \red{$N_{f}$} accounts for total size of the operator acting on the Choi state of \red{$f$}. Upon rearranging, we arrive at the form of Eq.~\eqref{eq:mainthmeq}.
\end{proof}

It is also worth noting that the bound presented is pessimistic in the sense that we took an absolute error from Lemma \ref{thm:shadowtomo} as standard deviation. 
Thus, we can expect the error to be considerably lower in practice so long as the additive error $\eta$ and ``failure'' probability $\gamma$ are small. 

\section{Error scaling in the low-shot regime} \label{sec:opt-fragment}

\begin{figure}[h]
    \begin{subfigure}{0.49\linewidth}
        \includegraphics[width=\linewidth]{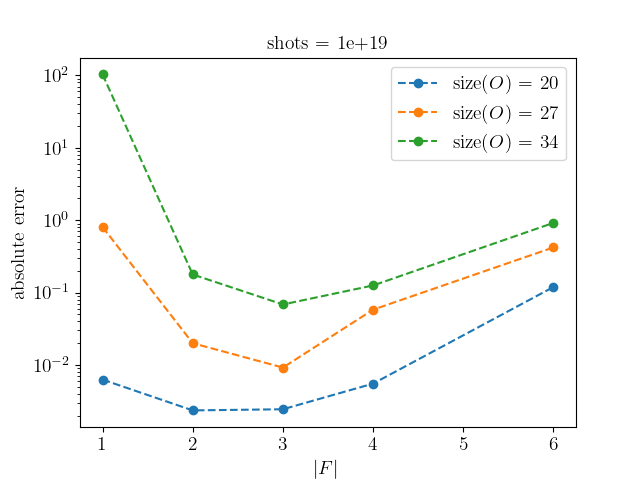}
        \caption{Theoretical scaling {predicted by} Theorem \ref{thm:mainthm} {for} a 36-qubit circuit and a total of $10^{19}$ shots.}
    \end{subfigure}
    \hfill
    \begin{subfigure}{0.49\linewidth}
        \includegraphics[width=\linewidth]{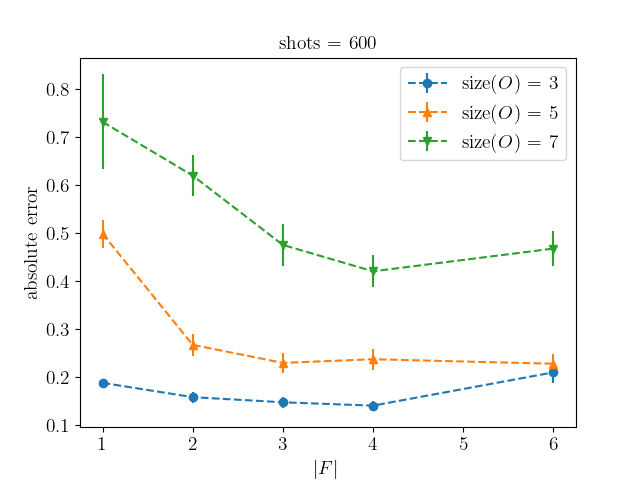}
        \caption{Empirical scaling obtained from simulations of a 12-qubit circuit and taking a total of 600 shots.}
    \end{subfigure}
    \caption{A comparison of \textbf{(left)} theoretical scaling and \textbf{(right)} empirical scaling with extremely low number of shots. We observe similar qualitative behavior where there exists an ``optimal'' number of fragments. Results shown for the cascade circuit in of Figure \ref{fig:cascade}.}
    \label{fig:opt-frag-fig}
\end{figure}

As briefly mentioned in Section \ref{sec:numerics}, the bound provided in Theorem \ref{thm:mainthm} is extremely pessimistic, and requires an unrealistic number of shots to estimate the expectation of any observable to sufficient accuracy. This is likely due to an overly conservative estimation on error propagation during tensor contraction. However, if shot noise is sufficiently high, tensor contraction will exponentially amplify the noise on each fragment. Thus, we should expect an increase in error as the number of fragments increase.

We tried replicating scenarios where shot noise for each fragment is significant by severely limiting the total number of shots. Meanwhile, we study the theoretical scaling in parameter regimes such that splitting large observables is significantly more favorable. As a result, there exhibits a clear separation among choices for the numbers of fragments. We show the results in Figure \ref{fig:opt-frag-fig}, where there is a surprising qualitative resemblance between theory and empirical evidence. Thus, while the theoretical results might provide some intuition about the behavior of circuit cutting in the high-weight observable, low shot count regime, more practical bounds should be developed for accurate prediction of error.

\end{document}